\documentclass[lnicst,sechang,a4paper]{svmultln}
\usepackage{amssymb}
\usepackage{amsmath}
\usepackage{amsfonts}
\usepackage{mathrsfs}
\usepackage{bbm}
\usepackage{todonotes}
\usepackage{color}
\usepackage{graphicx}
\usepackage{url}
\urldef{\maila}\path|{Remi.Bonnefoi, Lilian.Besson}@CentraleSupelec.fr|
\urldef{\mailb}\path|{Christophe.Moy,Jacques.Palicot}@CentraleSupelec.fr|
\urldef{\mailCS}\path|Firstname.Lastname@CentraleSupelec.fr|
\urldef{\mailc}\path|{Lilian.Besson,Emilie.Kaufmann}@Univ-Lille1.fr|
\urldef{\mailUL}\path|Firstname.Lastname@Univ-Lille1.fr|
\urldef{\mailCSUL}\path|First.Last@(CentraleSupelec,Univ-Lille1).fr|
\usepackage{enumitem}
\usepackage{stmaryrd}
\usepackage{epstopdf}
\usepackage{array}
\usepackage{multirow}
\usepackage[caption=false,font=footnotesize]{subfig}

\newcommand\Mark[1]{\textsuperscript{#1}}
\newcommand{\UCB}{\ensuremath{\mathrm{UCB}_1}}

\renewcommand{\Pr}{\mathbb{P}}

\begin{document}

\mainmatter  

\title{Multi-Armed Bandit Learning in IoT Networks:\\
Learning helps even in non-stationary settings
}
\titlerunning{Multi-Armed Bandit Learning in Non-Stationary IoT Networks}

\author{R\'emi Bonnefoi\Mark{1} \and Lilian Besson\Mark{1}\Mark{,2} \and\\
Christophe Moy\Mark{1} \and Emilie Kaufmann\Mark{2} \and Jacques Palicot\Mark{1}}  %
\authorrunning{R. Bonnefoi \and L. Besson \and C. Moy \and E. Kaufmann \and J. Palicot}

\institute{\Mark{1}CentraleSup\'elec (campus of Rennes), IETR, SCEE Team,\\
	               Avenue de la Boulaie - CS 47601, $35576$ Cesson-S\'evign\'e, France\\
    \Mark{2}Univ. Lille 1, CNRS, Inria, SequeL Team\\
		    UMR 9189 - CRIStAL,  F-59000 Lille, France\\
	\mailCSUL%
}


\maketitle

\begin{abstract}
Setting up the future Internet of Things (IoT) networks will require to support more and more communicating devices. We prove that intelligent devices in unlicensed bands can use Multi-Armed Bandit (MAB) learning algorithms to improve resource exploitation.
We evaluate the performance of two classical MAB learning algorithms, \UCB{} and Thomson Sampling, to handle the decentralized decision-making of Spectrum Access, applied to IoT networks; as well as learning performance with a growing number of intelligent end-devices.
We show that using learning algorithms does help to fit more devices in such networks, even when all end-devices are intelligent and are dynamically changing channel.
In the studied scenario, stochastic MAB learning provides a up to $16\%$ gain in term of successful transmission probabilities, and has near optimal performance even in non-stationary and non-\emph{i.i.d.} settings with a majority of intelligent devices.

\keywords{Internet of Things, Multi-Armed Bandits, Reinforcement Learning, Cognitive Radio, Non-Stationary Bandits.}
\end{abstract}

\section{Introduction}

Unlicensed bands are more and more used and considered for mobile and LAN communication standards (WiFi, LTE-U), and for Internet of Things (IoT) standards for short-range (ZigBee, Z-Wave, Bluetooth) and long-range (LoRaWAN, SIGFOX, Ingenu, Weightless) communications \cite{Centenaro16}.
This heavy use of unlicensed bands will cause performance drop, and could even compromise IoT promises.

Efficient Medium Access (MAC) policies allow devices to avoid interfering traffic and can significantly reduce the spectrum contention problem in unlicensed bands.
As end-devices battery life is a key constraint of IoT networks,
this leads to IoT protocols using as low signaling overhead as possible and simple ALOHA-based mechanisms.
In this article, we analyze the performance of Multi-Armed Bandits (MAB) algorithms \cite{LaiRobbins85,bubeck2012regret}, used in combination with a time-frequency slotted ALOHA-based protocol.
We consider the Upper-Confidence Bound (\UCB) \cite{Auer}, and the Thompson-Sampling (TS) algorithms \cite{Thompson33,AgrawalGoyal11,
Kaufmann12}.

MAB learning has already been proposed in Cognitive Radio (CR) \cite{Haykin}, and in particular, for sensing-based Dynamic Spectrum Access (DSA) in licensed bands \cite{Jouini}.
Recently, TS and \UCB{} algorithms have been used for improving the spectrum access in (unlicensed) WiFi networks \cite{Toldov}, and the \UCB{} algorithm was used in a unlicensed and frequency- and time-slotted IoT network \cite{Bonnefoi}.
Many recent works show that MAB algorithms work well for real-world radio signal.
However, even with only one dynamic user using the learning algorithm, the background traffic or the traffic of the other devices is never really stationary or \emph{i.i.d} (independent and identically distributed).
In recent works like \cite{Bonnefoi}, several devices are using bandit algorithms, and the assumptions made by the stochastic bandit algorithms are not satisfied: as several agents learn simultaneously, their behavior is neither stationary nor \emph{i.i.d}.
As far as we know, we provide the first study to confirm robustness of the use of stochastic bandit algorithms for decision making in IoT networks with a large number of intelligent devices in the network, which makes the environment not stationary at all, violating the hypothesis required for mathematical proofs of bandit algorithms convergence and efficiency.

The aim of this article is to assess the potential gain of learning algorithms in IoT scenarios, even when the number of intelligent devices in the network increases, and the stochastic hypothesis is more and more questionable.
To do that, we suppose an IoT network made of two types of devices: static devices that use only one channel (fixed in time), and dynamic devices that can choose the channel for each of their transmissions. Static devices form an interfering traffic, which could have been generated by devices using other standards as well.
We first evaluate the probability of collision if dynamic devices randomly select channels (naive approach), and if a centralized controller optimally distribute them in channels (ideal approach).
Then, these reference scenarios allow to evaluate the performance of \UCB{} and TS algorithms in a decentralized network, in terms of successful communication rate, as it reflects the network efficiency.

The rest of this article is organized as follows. The system model is introduced in Section 2. Reference policies are described in Section 3, and MAB algorithms are introduced in Section 4. Numerical results are presented in Section 5.

\section{System model and notations}

\begin{figure}[!t]
\centering
\includegraphics[scale=0.28]{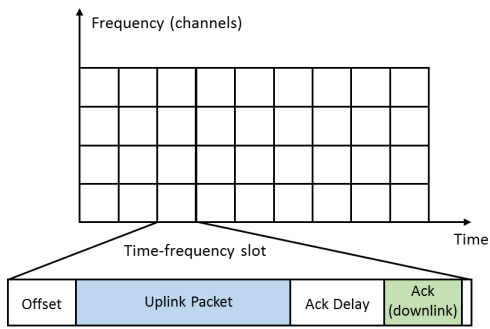}
\caption{The considered time-frequency slotted protocol. Each frame is composed by a fix-duration uplink slot in which the end-devices transmit their packets. If a packet is well received, the base station replies by transmitting an \emph{Ack}, after the ack delay.}
\label{fig:protocol}
\vspace*{-15pt}
\end{figure}

As illustrated in Figure \ref{fig:protocol}, we suppose a slotted protocol.
All devices share a synchronized time, and know in advance the finite number of available RF channels.
In each time slot, devices try to send packets to the unique Base Station, which listens continuously to all channels, following an ALOHA-based communication (no sensing).
Each time slot is divided in two parts: first for uplink communications in which data packets are sent by end-devices to the base station. If only one packet is sent in this part of the slot, the base station can decode it and sends an acknowledgement to the device in the second part.
If two or more devices send an uplink packet in the same slot, the uplink packets collide (\emph{i.e.}, there is a \emph{collision}), and the acknowledgement \emph{Ack} is not transmitted.
This way, no collision can occur on the downlink messages, easing the analysis of collisions.

There are two types of end-devices in the network:
\begin{itemize}
\item \emph{Static} end-devices have poor RF abilities, and each of them uses only one channel to communicate with the base station. Their choice is assumed to be fixed in time (stationary) and independent (\emph{i.i.d.}). The traffic generated by these devices is considered as an interfering traffic for other devices.
\item \emph{Dynamic} (or \emph{smart}) end-devices have richer RF abilities, they can use all the available channels, by quickly reconfiguring their RF transceiver on the fly. They can also store communication successes or failures they experienced in each channel, in order to change channel, possibly at every time slot.
\end{itemize}

There are $N_c \geq 1$ channels, $D \geq 0$ dynamic end-devices, and $S \geq 0$ static devices.
Furthermore, in channel $i \in \llbracket 1; N_c \rrbracket$ there are $0 \leq S_i \leq S$ static devices (so $S = \sum_{i=1}^{N_c} S_i$).
We focus on \emph{dense networks}, in which the number of devices $S + D$ is very large compared to $N_c$ (about $1000$ to $10000$, while $N_c$ is about $10$ to $50$).
As this problem is only interesting if devices are able to communicate reasonably efficiently with the base station, we assume devices only communicate occasionally, \emph{i.e.}, with a low \emph{duty cycle}, as it is always considered for IoT.

We suppose that all devices follow the same emission pattern, being fixed in time, and we choose to model it as a simple Bernoulli process:
all devices have the same probability to send a packet in any (discrete) temporal slot, and we denote $p \in (0, 1)$ this probability\footnote{In the experiments below, $p$ is about $10^{-3}$, because in a crowded network $p$ should be smaller than $N_c / (S + D)$ for all devices to communicate successfully (in average).}.
The parameter $p$ essentially controls the frequency of communication for each device, once the time scale is fixed (\emph{i.e.}, real time during two messages), and $1/p$ is proportional to the \emph{duty cycle}.

The goal is to design a simple sequential algorithm, to be applied identically by each dynamic device, in a fully distributed setting (each device runs its own algorithm, from its observations), in order to minimize collisions and maximize the fraction of successful transmissions of all the dynamic devices.

Before explaining how this goal presents similarity with a \emph{multi-armed bandit problem}, we present some natural baseline policies (\emph{i.e.}, algorithms).

\section{Three reference policies}

This section presents three different policies that will be used to assess the efficiency of the learning algorithms presented in the next section.
The first one is naive but can be used in practice, while the two others are very efficient but require full knowledge on the system (i.e., an oracle) and are thus unpractical.

\subsection{Naive policy: Random Channel Selection}

We derive here the probability of having a successful transmission, for a dynamic device, in the case where all the dynamic devices make a purely random channel selection (i.e., uniform on $i \in \llbracket 1; N_c \rrbracket = \{1, \dots, N_c\}$).

In this case, for one dynamic device, a successful transmission happens if it is the only device to choose channel $i$, at that time slot.
The probability of successful transmission is computed as follows, because the $S_i$ static devices in each channel $i$ are assumed to be independent, and static and dynamic devices are assumed to \emph{not} transmit at each time $t$ with a fixed probability $1-p$ :
\vspace*{-5pt}
\begin{equation}
\Pr(\text{success}|\text{sent}) = \sum_{i=1}^{N_c} \underbrace{\Pr(\text{success}|\text{sent in channel}\;i)}_{\text{No one else sent in channel}\; i} \; \underbrace{\Pr(\text{sent in channel}\,i)}_{= 1/N_c, \text{by uniform choice}} \end{equation}
All dynamic devices follow the same policy in this case, so the probability of transmitting at that time in channel $i$ for any dynamic device is $p / N_c$, and there are $D-1$ other dynamic devices.
As they are independent, the probability that no other dynamic device sent in $i$
is $q = \Pr(\bigcap_{k=1}^{D-1} \text{device}\;k\;\text{did not sent in}\;i) = \prod_{k=1}^{D-1} \Pr(\text{device}\;k\;\text{did not sent in}\;i)$. And $\Pr(\text{device}\;k\;\text{sent in}\;i) = p \times 1 / N_c$, by uniform choice on channels and the Bernoulli emission hypothesis. So $q = \prod_{k=1}^{D-1} (1 - p/N_c) = (1-p/N_c)^{D-1}$. Thus we can conclude,
\vspace*{-6pt}
\begin{align}\label{eq:strategynaive}
\Pr(\text{success}|\text{sent})
 & = \sum_{i=1}^{N_c} \underbrace{(1 - p / N_c)^{D-1}}_{\text{No other dynamic device}} \times \underbrace{(1-p)^{S_i}}_{\text{No static device}} \times\; \frac{1}{N_c} \nonumber \\
 & = \frac{1}{N_c} \left(1-\frac{p}{N_c}\right)^{D-1} \sum_{i=1}^{N_c} (1-p)^{S_i} .
\end{align}
This expression \eqref{eq:strategynaive} is constant (in time), and easy to compute numerically, but comparing the successful transmission rate of any policy against this naive policy is important, as any efficient learning algorithm should outperform it.

\subsection{(Unachievable) Optimal oracle policy}

We investigate in this section the optimal policy that can be achieved if the dynamic devices have a perfect knowledge of everything, and a fully centralized decision making\footnote{This optimal policy needs an \emph{oracle} seeing the entire system, and affecting all the dynamic devices, once and for all, in order to avoid any signaling overhead.} is possible.
We want to find the stationary repartition of devices into channels that maximizes the probability of having a successful transmission.

If the oracle draws once uniformly at random a configuration of dynamic devices, with $D_i$ devices affected to channel $i$ is fixed (in time, i.e., stationary),
then this probability is computed as before:
\vspace*{-5pt}
\begin{align}
\label{eq:prob_col}
\Pr(\text{success}|\text{sent})
& = \sum_{i=1}^{N_c} \Pr(\text{success}|\text{sent in channel}\;i) \; \Pr(\text{sent in channel}\;i) \nonumber \\
& = \sum_{i=1}^{N_c} \underbrace{(1 - p)^{D_i - 1}}_{\;\;D_i - 1 \;\text{others}\;\;} \times \underbrace{(1 - p)^{S_i}}_{\;\;\text{No static device}\;\;} \times \underbrace{ D_i / D }_{\;\;\text{Sent in channel}\; i\;\;}.
\end{align}

Consequently, the optimal allocation vector $(D_1,\dots,D_{N_c})$ is the solution of the following real-valued constraint optimization problem :
\vspace*{-5pt}
\begin{subequations}
\label{eq:prob}
\begin{align}
\underset{D_1,\dots,D_{N_c}}{\arg\max}\; & \sum_{i=1}^{N_c} D_i (1 - p)^{S_i + D_i -1}, \label{eq:optPb}\\
\text{such that}\;\; & \sum_{i=1}^{N_c} D_i = D, \label{eq:eqCstr}\\
& D_i \geq 0 \qquad \forall i\in\llbracket 1;N_c\rrbracket . \label{eq:ineqCstr}
\end{align}
\end{subequations}

\begin{proposition}\label{prop:Lagrangian}
The \emph{Lagrange multipliers} method \cite{BoydVanderberghe04} can be used to solve the constraint real-valued maximization problem introduced in equation \eqref{eq:prob}.

It gives a closed form expression for the optimal solution $D_i^*(\lambda)$, depending on the system parameters, and the unknown Lagrange multiplier $\lambda \in \mathbb{R}$. 
\begin{equation}\label{eq:Dilambda}
D_i^*(\lambda) = \left(\frac{1}{\log(1-p)}\left[ \mathcal{W}\left(\frac{\lambda e}{(1-p)^{S_i-1}} \right)-1 \right]\right)^{\dag} .
\end{equation}
\end{proposition}
\begin{proof}
\begin{itemize}
    \item
    In a realistic scenario, we can assume that $D_i\leq \frac{-2}{\ln\left(1-p\right)} \approx \frac{2}{p},\quad \forall i\in\llbracket 1;N_c \rrbracket$. For such values for $D_i$, the objective function $f: (D_1, \dots, D_{N_c}) \mapsto \sum_{i=1}^{N_c} D_i (1 - p)^{S_i + D_i -1}$ is concave as the sum of concave functions
    \footnote{It worth noting that $f$ is neither concave nor quasi-concave on $[0,\infty)^{N_c}$ \cite{Luenberger68,Yaari77}.}.
        \item
The Lagrange multipliers method can be applied to the optimization problem \eqref{eq:optPb}, with a concave objective function $f$, linear equality constraints \eqref{eq:eqCstr} and linear inequality constraints \eqref{eq:ineqCstr}. The strong duality condition is satisfied in this case \cite{BoydVanderberghe04}, so finding the saddle points will be enough to find the maximizers.
    \hfill{}$\square$
    \end{itemize}
\end{proof}

Where $(a)^{\dag} = \max(a,0)$, and $\mathcal{W}$ denotes the $\mathcal{W}$-Lambert function which is the reciprocal bijection of $x \mapsto xe^x$ on $\mathbb{R^+} = [0, +\infty)$ \cite{Corless96}.
Moreover, condition \eqref{eq:eqCstr} implies that the Lagrange multiplier $\lambda$ is the solution of the constraint
\begin{equation}\label{eq:constraintLambda}
 \sum_{i=1}^{N_c} D_i^*(\lambda) = D.
\end{equation}

Equation \eqref{eq:constraintLambda} can be solved numerically, with simple one-dimensional root finding algorithms.
Solving the optimization problem provides the optimal real number value for $D_i^*$, which has to be rounded to find the optimal number of devices for channel $i$ :
$\widehat{D_i} = \lfloor D_i^* \rfloor$ for $1 \leq i < N_c$, and $\widehat{D_{N_c}} = D - \sum_{i=1}^{N_c - 1} \widehat{D_i}$.

\subsection{A greedy approach of the oracle strategy}

We propose a \emph{sequential} approximation of the optimal policy:
the third solution is a sub-optimal naive policy, simple to set up, but also unpractical as it also needs an oracle.
End-devices are iteratively inserted in the channels with the lowest load (\emph{i.e.}, the index $i$ minimizing $S_i + D_i(\tau)$ at global time step $\tau$). Once the number of devices in each channel is computed, the probability of sending successfully a message is also given by equation \eqref{eq:prob_col}.
This is the policy that would be used by dynamic devices if they were inserted one after the other, and if they had a perfect knowledge of the channel loads.

\vspace*{-5pt}
\section{Sequential policies based on bandit algorithms}

We now present the stochastic Multi-Armed Bandit (MAB) model,
and the two stochastic MAB algorithms used in our experiments \cite{bubeck2012regret}.
While the stochastic MAB model has been used to describe some aspects of Cognitive Radio systems, it is in principle not suitable for our IoT model, due to the non-stationarity of the channels occupancy caused by the learning policy used by dynamic objects.


\subsection{Stochastic Multi-Armed Bandits}

A Multi-Armed Bandit problem is defined as follows \cite{Thompson33,Robbins52,LaiRobbins85}.
There is a fixed number $N_c \geq 1$ of levers, or ``arms'', and a player has to choose one lever at each discrete time $t \geq 1, t \in \mathbb{N}$, denoted as $A(t) = k \in\{1,\dots,N_c\}$.
Selecting arm $k$ at time $t$ yields a (random) \emph{reward}, $r_k(t) \in \mathbb{R}$, and the goal of the player is to maximize the sum of his rewards, $r_{1 \dots T} = \sum_{t = 1}^T r_{A(t)}(t)$.

A well-studied version of this problem is the so-called ``stochastic'' MAB, where the sequence of rewards drawn from a given arm $k$ is assumed to be independent and identically distributed (\emph{i.i.d}) under some distribution $\nu_k$, that has a mean $\mu_k$. Several types of reward distributions have been considered, for example distributions that belong to a one-dimensional exponential family (\emph{e.g.}, Gaussian, Exponential, Poisson or Bernoulli distributions).
We consider Bernoulli bandit models, in which $r_k(t) \sim \mathrm{Bern}(\mu_k)$, that is, $r_k(t) \in \{0,1\}$ and $\mathbb{P}(r_k(t) = 1) = \mu_k$.

The problem parameters $\mu_1,\dots,\mu_K$ are unknown to the player, so to maximize his cumulated rewards, he must learn the distributions of the channels, to be able to progressively focus on the best arm (\emph{i.e.}, the arm with largest mean).
This requires to tackle the so-called \emph{exploration-exploitation dilemma}: a player has to try all arms a sufficient number of times to get a robust estimate of their qualities, while not selecting the worst arms too many times.

In a Cognitive Radio application, arms model the \emph{channels}, and players are the \emph{dynamic end-devices}.
For example in the classical OSA setting with sensing \cite{Jouini}, a single dynamic device (a player) sequentially tries to access channels (the arms), and collects a reward of 1 if the channel is available 
and 0 otherwise. So rewards represent the \emph{availability} of channels, and
the parameter $\mu_k$ represents the mean availability of channel $k$.

Before discussing the relevance of a multi-armed bandit model for our IoT application, we present two bandit algorithms, UCB1 and Thompson Sampling,
which both strongly rely on the assumption that rewards are \emph{i.i.d.}.

\subsection{The \UCB{} algorithm}

A naive approach could be to use an empirical mean estimator of the rewards for each channel, and select the channel with highest estimated mean at each time. 
This greedy approach is known to fail dramatically \cite{LaiRobbins85}. Indeed, with this policy, the selection of arms is highly dependent on the first draws, if the first transmission in a channel fails, the device will never use it again.
Rather than relying on the empirical mean reward, Upper Confidence Bounds algorithms
instead use a \emph{confidence interval} on the unknown mean $\mu_k$ of each arm,
which can be viewed as adding a ``bonus'' exploration to the empirical mean.
They follow the ``\emph{optimism-in-face-of-uncertainty}'' principle : at each step, they play according to the best model,
as the statistically best possible arm (\emph{i.e.}, the highest upper confidence bound) is selected.

More formally, for one device, let $N_k(t) = \sum_{\tau=1}^t \mathbbm{1}(A(\tau) = k)$ be the number of times channel $k$ was selected up-to time $t \geq 1$.
The empirical mean estimator of channel $k$ is defined as the mean reward obtained by selecting it up to time $t$, $\widehat{\mu_k}(t) = 1 / N_k(t) \sum_{\tau=1}^t r_k(\tau) \mathbbm{1}(A(\tau) = k) $.
For \UCB, the \emph{confidence} term is $B_k(t) = \sqrt{\alpha \log(t) / N_k(t)}$,
giving the upper confidence bound $U_k(t) = \widehat{\mu_k}(t) + B_k(t)$, which is used by the device to decide the channel for communicating at time step $t+1$: $A(t+1) = \arg\max_{1\leq k \leq N_c} U_k(t)$.
\UCB{} is an \emph{index policy}.

The \UCB{} algorithm uses a parameter $\alpha > 0$, originally, $\alpha$ set to $2$ \cite{Auer}, but empirically $\alpha = 1/2$ is known to work better (uniformly across problems), and $\alpha > 1/2$ is advised by the theory \cite{bubeck2012regret}.
%
In our model, every dynamic device implements its own \UCB{} algorithm, \emph{independently}. For one device, the time $t$ is the number of time it accessed the network (following its Bernoulli transmission process,\emph{i.e.}, its duty cycle), \emph{not} the total number of time slots from the beginning, as rewards are only obtained after a transmission, and IoT objects only transmit sporadicly, due to low transmission duty cycles.

\subsection{Thompson Sampling}

Thompson Sampling \cite{Thompson33} was introduced in 1933 as the very first bandit algorithm, in the context of clinical trials (in which each arm models the efficacy of one treatment across patients). Given a prior distribution on the mean of each arm, the algorithm selects the next arm to draw based on samples from the \emph{conjugated} posterior distribution, which for Bernoulli rewards is a Beta distribution.

A Beta prior $\mathrm{Beta}(a_k(0)=1,b_k(0)=1)$ (initially uniform) is assumed on $\mu_k \in [0, 1]$, and at time $t$ the posterior is $\mathrm{Beta}(a_k(t),b_k(t))$.
After every channel selection, the posterior is updated to have $a_k(t)$ and $b_k(t)$ counting the number of successful and failed transmissions made on channel $k$.
So if the \emph{Ack} message is received, $a_k(t+1) = a_k(t) + 1$, and $b_k(t+1) = b_k(t)$, otherwise $a_k(t+1) = a_k(t)$, and $b_k(t+1) = b_k(t) + 1$.
Then, the decision is done by \emph{sampling} an \emph{index} for each arm, at each time step $t$, from the arm posteriors: $X_k(t) \sim \mathrm{Beta}(a_k(t), b_k(t))$, and the chosen channel is simply the channel $A(t+1)$ with highest index $X_k(t)$. For this reason, Thompson Sampling is a \emph{randomized index policy}.

Thompson Sampling, although being very simple, is known to perform well for stochastic problems, for which it was proven to be asymptotically optimal \cite{AgrawalGoyal11,Kaufmann12}.
It is known to be empirically efficient, and for these reasons it has been used successfully in various applications, including on problems from Cognitive Radio \cite{Toldov,Mitton}, and also in previous work on decentralized IoT-like networks \cite{Darak16}.

\subsection{A bandit model for IoT}

Our IoT application is challenging in that there are \emph{multiple} players (the dynamic devices) interacting with the \emph{same} arms (the channels), without any centralized communication (they do not even know the total number of dynamic devices).

Considered alone, each dynamic device implements a learning algorithm to play a bandit game, the device is consequently a smart device. In each time slot, if it has to communicate (which happens with probability $p$), then it chooses a channel and it receives a reward $1$ if the transmission is successful, $0$ otherwise.
Each device aims at maximizing the sum of the rewards collected during its communication instants, which shall indeed maximize the fraction of successful transmissions. Besides the modified time scale (rewards are no longer collected at every time step), this looks like a bandit problem.
However, it cannot be modeled as a stochastic MAB, as the rewards are clearly \emph{not} \emph{i.i.d}: they not only depend on the (stationary, \emph{i.i.d}) behavior of the static devices, but also on the behavior of other smart devices, that is not stationary (because of learning).

Despite this, we show in the next section that running a stochastic bandit algorithm for each device based on its own rewards is surprisingly successful.

\subsubsection{Multi-Player MAB with \emph{collision avoidance}?}
Another idea could be to try to use a \emph{multi-player MAB} model, as proposed by \cite{Zhao10}, to describe our problem.

In that case, the static and dynamic devices effect is decoupled, and arms only model the availability of the channels in the absence of dynamic devices : they are \emph{i.i.d.} with mean $\mu_i = 1 - p S_i$.
Moreover, dynamic devices are assumed to be able to \emph{sense} a channel before sending \cite{Zhao10}, and so communicate only if no static device is detected on the channel.
The smart devices try to learn the arms with highest means, while coordinating to choose different arms, \emph{i.e.}, avoid collisions in their choice, in a decentralized manner.
However, in this model it is assumed that the multiple agents can know that they experienced a collision with another agent, which is non-realistic for our problem at stake, as our model of smart device cannot do sensing nor differentiate collisions between smart and non-smart devices.


\subsubsection{\emph{Adversarial} bandit algorithms?}
Instead of using MAB algorithms assuming a stochastic hypothesis on the system, we could try to use MAB algorithms designed to tackle a more general problem, that makes no hypothesis on the interfering traffic.
The \emph{adversarial MAB} algorithms is a broader family, and a well-known and efficient example is the $\mathrm{Exp}3$ algorithm \cite{bubeck2012regret}.
Empirically, the $\mathrm{Exp}3$ algorithm turned out to perform worse than both \UCB{} and TS in the same experiments.
Contrarily to the two stochastic algorithms, the use of $\mathrm{Exp}3$ is correctly justified, even in the non-stationary and non-\emph{i.i.d}, as its performance guarantee are true in \emph{any} setting.
But it is not so surprising that it performs worse, as the theoretical performance guarantees of adversarial MAB algorithms are an order of magnitude
worse than the one for stochastic ones.
More is left on this aspect for our future work.

\section{Experiments and numerical results}\label{sec:Experiments}

\begin{figure*}[!t]
\centering
\subfloat[10\% of smart devices]{\includegraphics[scale=0.4]{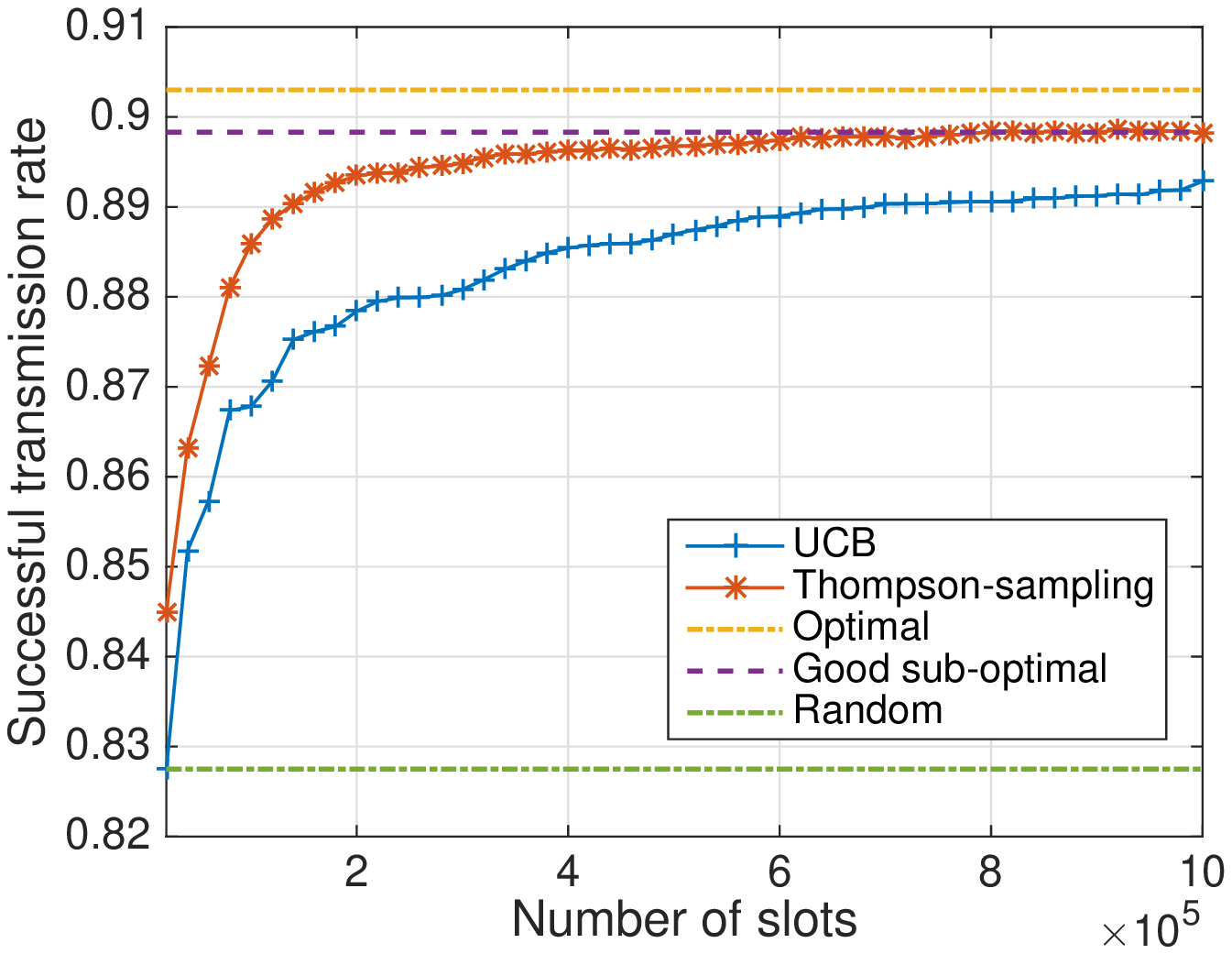}
\label{fig:10intelligent}}
\hfill
\subfloat[30\% of smart devices]{\includegraphics[scale=0.4]{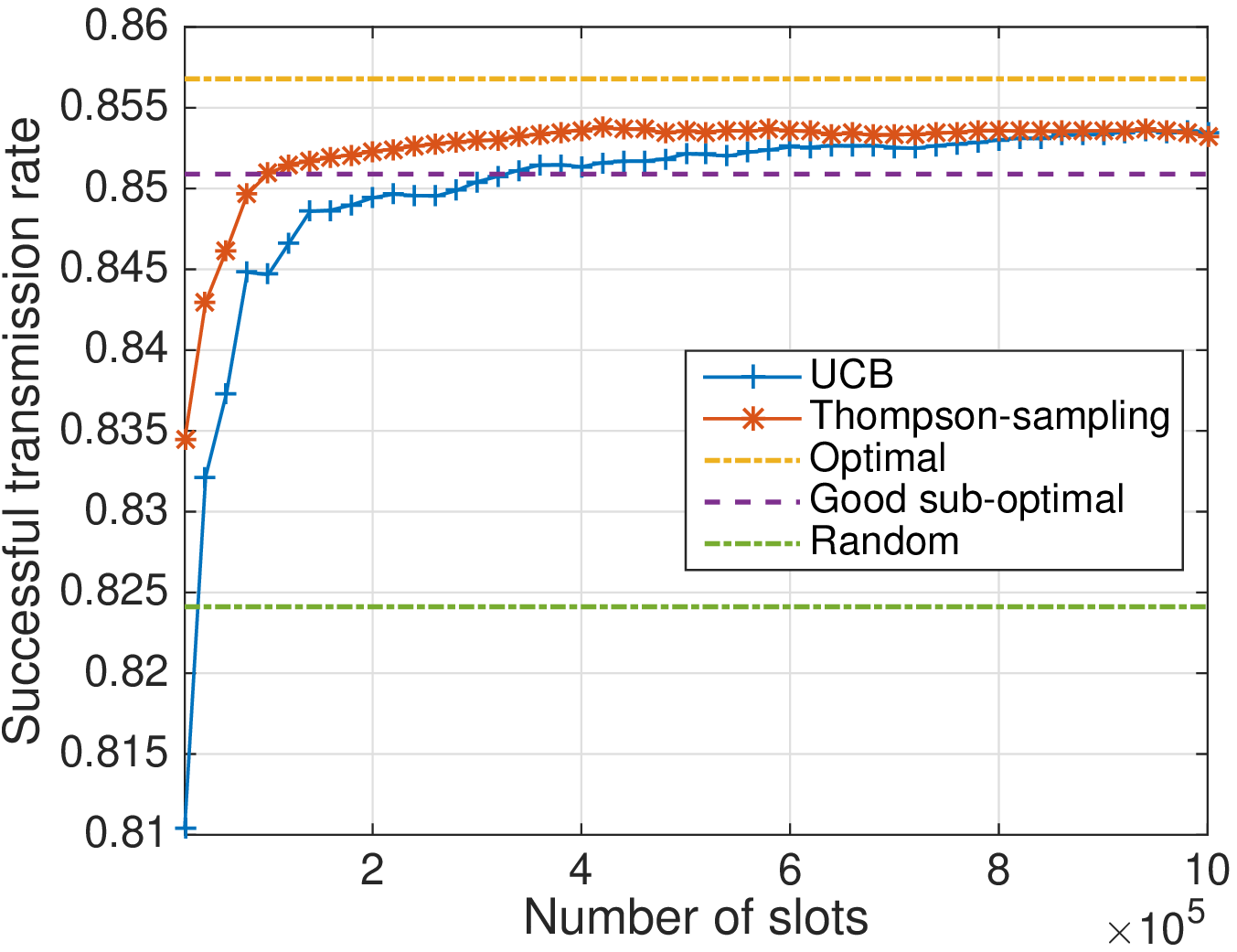}
\label{fig:30intelligent}}
\vskip\baselineskip
\vspace*{-20pt}
\subfloat[50\% of smart devices]{\includegraphics[scale=0.4]{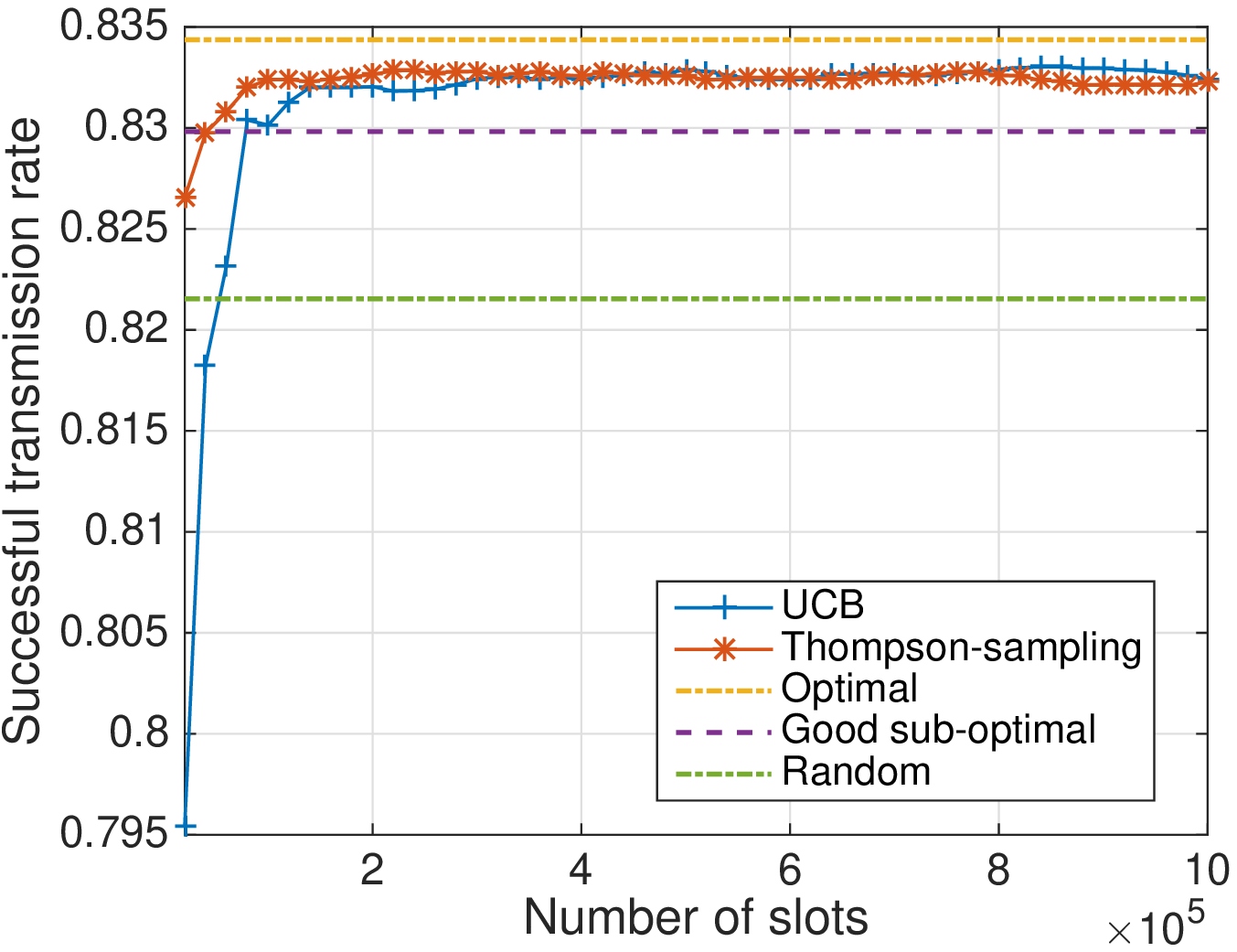}
\label{fig:50intelligent}}
\hfill
\subfloat[100\% of smart devices]{\includegraphics[scale=0.4]{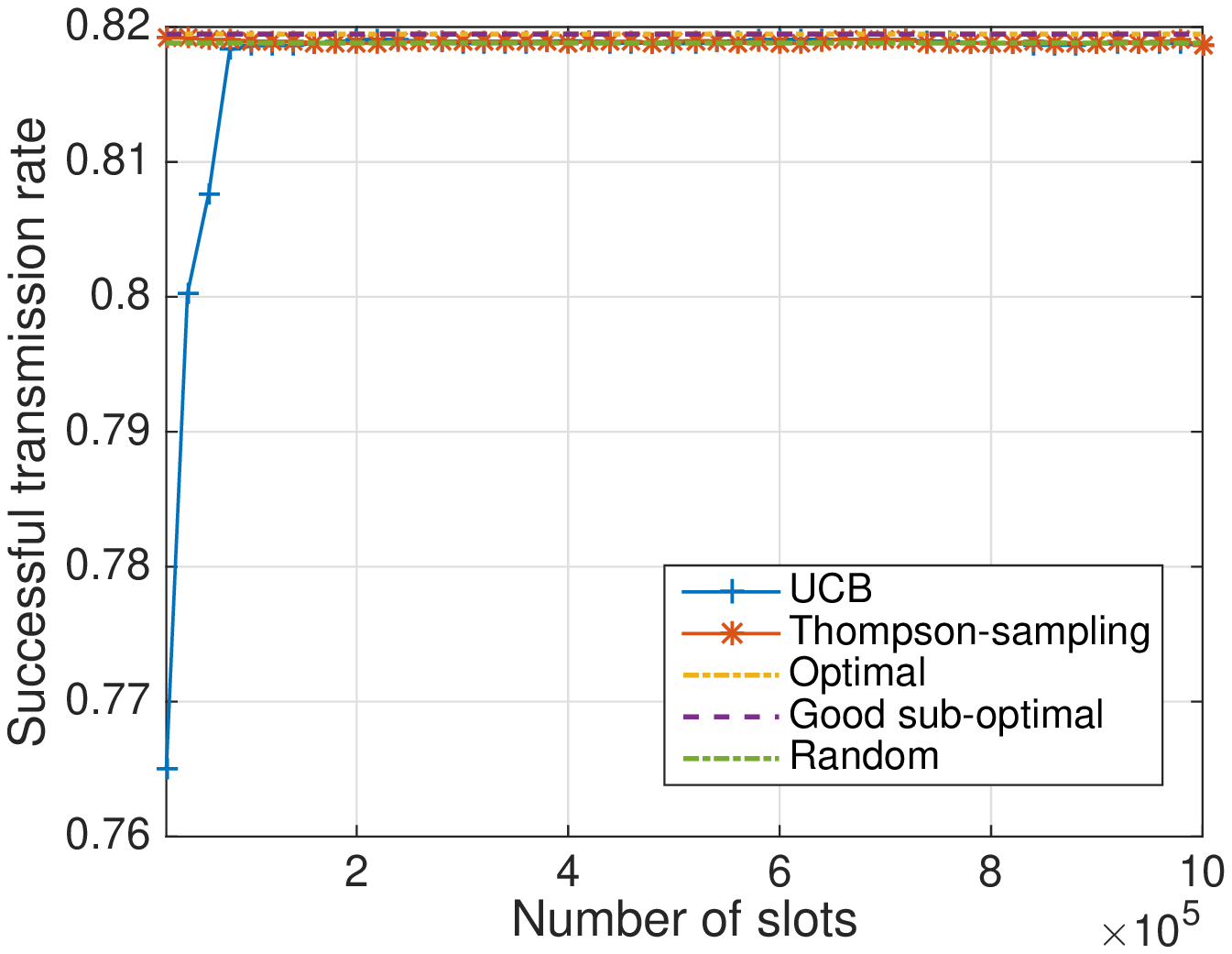}
\label{fig:100intelligent}}
\caption{Performance of $2$ MAB algorithms (\UCB{} and Thompson Sampling), compared to extreme references without learning or oracle knowledge, when the proportion of smart end-devices in the network increases, from $10\%$ to $100\%$ (limit scenario).}
\label{fig:from10to100}
\vspace*{-10pt}
\end{figure*}

We suppose a network with $S + D = 2000$ end-devices, and one IoT base station.
Each device sends packets following a Bernoulli process, of probability $p = 10^{-3}$ (\emph{e.g.}, this is realistic: one packet sent about every $20$ minutes, for time slots of $1\mathrm{s}$).
The RF band is divided in $N_c = 10$ channels\footnotemark.
Each static device only uses one channel, and their uneven repartition\footnotemark[\value{footnote}] in the $10$ channels is: $(S_1,\cdots, S_{N_c}) = S\times(0.3, \, 0.2, \, 0.1, \, 0.1, \, 0.05, \, 0.05, \, 0.02, \, 0.08, \, 0.01,$ $0.09)$, to keep the same proportions when $S$ decreases. The dynamic devices have access to all the channels, and use learning algorithms.
%
We simulate the network during $10^6$ discrete time slots, during which each device transmits on average $1000$ packets (i.e., the learning time is about $1000$ steps, for each algorithm).
\footnotetext{We tried similar experiments with other values for $N_c$ and this repartition vector, and results were similar for non-homogeneous repartitions. Clearly, the problem is less interesting for homogeneous repartition, as all channels appear the same for dynamic devices, and so even with $D$ small in comparison to $S$, the system behaves like in Fig.\ref{fig:100intelligent}, where the performance of the five approaches are very close.}

Figure \ref{fig:from10to100} presents the evolution of the successful transmission rate, as a function of time. 
The two MAB algorithms, \UCB{} and Thompson Sampling (TS), are compared against the naive random policy from below, and the two oracle policies (optimal and greedy) from above.
The results are displayed when $10$, $30$, $50$ and $100\%$ of the traffic is generated by dynamic devices.

We can see in Figure \ref{fig:from10to100} that the TS algorithm (in red) outperforms the \UCB{} algorithm (in blue), when the number of end-devices is below 50\%. When the number of end-devices is higher, both algorithms have almost the same performance, and perform well after very few transmissions (quick convergence).
Moreover, we can see in Figures \ref{fig:10intelligent}, \ref{fig:30intelligent}, and \ref{fig:50intelligent} that both have better success rate than the random policy and the probability of successful transmission is between the oracle optimal and oracle suboptimal policies.
For instance, for $10\%$ of dynamic devices, after about $1000$ transmissions, using \UCB{} over the naive uniform policy improved the successful transmission rate from $83\%$ to $88\%$, and using Thompson Sampling improved it to $89\%$.
Increasing the number of end-devices decreases the gap between the optimal and random policies: the more dynamic devices, the less useful are learning algorithms, and basically for networks with only dynamic devices, the random policy is as efficient as the optimal one, as seen in Figures \ref{fig:100intelligent} and \ref{fig:perf_learning}.

To better assess the evolution of the optimal policy compared to the random one, we have displayed on Figure \ref{fig:perf_learning} the evolution of the gain, in term of successful transmissions rate, provided by the optimal oracle and the two learning policies, after $10^6$ time slots, \emph{i.e.}, about $1000$ transmissions for each object.
We can see that when the proportion of end-devices is low (\emph{e.g.}, $1\%$ of devices are dynamic), the optimal policy provides an improvement of $16\%$ compared to random channel selection.
The TS algorithm always provides near-optimal performance, but the \UCB{} algorithm has a lowest rate of convergence and performs consequently worse after $1000$ transmissions, for instance it only provides a gain of $12\%$ for the same proportion of dynamic devices ($1\%$).
\vspace*{-3pt}

\begin{figure}[!t]
\centering
\includegraphics[scale=0.41]{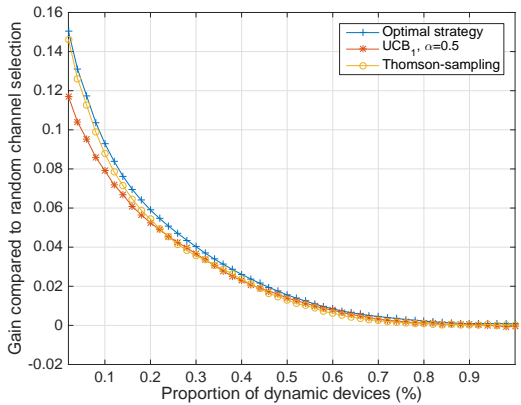}
\caption{Learning with \UCB{} and Thomson Sampling, with many smart devices.}
\label{fig:perf_learning}
\vspace*{-5pt}
\end{figure}
Figure \ref{fig:perf_learning} also shows that learning keeps near-optimal performance even when the proportion of devices becomes large.
Note that when this proportion increases, the assumptions of a stochastic MAB model are clearly violated, and there is no justification for the efficiency of TS and \UCB{} algorithms.
Hence, it is surprising to have near optimal performance with stochastic MAB algorithms applied to partly dynamic and fully dynamic scenarios.

\vspace*{-10pt}
\section{Conclusion}

In this article, we proposed an evaluation of the performance of MAB learning algorithms in IoT networks,
with a focus on the convergence of algorithms, in terms of successful transmission rates, when the proportion of intelligent dynamic devices changes.
Concretely, increasing this probability allows to insert more objects in the same network, while maintaining a good Quality of Service.
We show that \UCB{} and TS have near-optimal performance, even when their underlying \emph{i.i.d.} assumption is violated by the many ``intelligent'' end-devices.

This is both a surprising and a very encouraging result, showing that application of bandit algorithms tailored for a stochastic model is still useful in broader settings.
The fully \emph{decentralized} application of classic stochastic MAB algorithms are almost as efficient as the best possible centralized policy in this setting, after a short learning period, even though the dynamic devices \emph{can not} communicate with each others, and \emph{do not} know the system parameters.
We will investigate this behavior in order to understand it better theoretically.
We will also experiment more with adversarial algorithms, to confirm that they work less efficiently than stochastic bandit algorithms in our non-stochastic setting.

Moreover, for sake of simplicity we supposed that all devices use the same standard. Our future work will consider more realistic interference scenarios and IoT networks, with, \emph{e.g.}, non-slotted time, more than one base station etc.

\begin{small}
\subsubsection{Acknowledgements}
This work is supported by the French National Research Agency (ANR), under the projects SOGREEN (grant coded: \texttt{N ANR-14-CE28-0025-02}) and BADASS (\texttt{N ANR-16-CE40-0002}), by R\'egion Bretagne, France, by the French Ministry of Higher Education and Research (MENESR) and ENS Paris-Saclay.\newline

\emph{Note}: the simulation code used for the experiments in Section \ref{sec:Experiments} is for MATLAB or GNU Octave,
and is open-sourced under the MIT License,\newline at:
\verb|https://Bitbucket.org/scee_ietr/rl_slotted_iot_networks|.
\end{small}
\vspace{-10pt}
\bibliographystyle{ieeetr}
\bibliography{biblio_RIoT}

\end{document}